\documentclass[%
reprint,
showpacs,
amsmath,amssymb,
aps,
prb,
dvipdfm
]{revtex4-1}
\usepackage{graphicx}
\usepackage{dcolumn}
\usepackage{bm}
\usepackage{color}

\usepackage{amsthm}
\usepackage{amsmath}	
  \usepackage{mathtools}

\begin{document}

\newcommand{\BigFig}[1]{\parbox{12pt}{\Huge #1}}
\newcommand{\BigZero}{\BigFig{0}}

\renewcommand{\topfraction}{0.99}
\renewcommand{\bottomfraction}{0.99}
\renewcommand{\dbltopfraction}{0.99}
\renewcommand{\textfraction}{0.1}
\renewcommand{\floatpagefraction}{0.99}
\renewcommand{\dblfloatpagefraction}{0.99}

 \newtheorem{theo}{Theorem}

\def\U#1{{%
\def\O{\mbox{O}}
\def\u{\mbox{u}}
\mathcode`\u=\mu
\mathcode`\O=\Omega
\mathrm{#1}}}
\def\ii{{\mathrm{i}}}
\def\jj{\,\mathrm{j}}                   
\def\ee{{\mathrm{e}}}
\def\dd{{\mathrm{d}}}
\def\cc{{\mathrm{c.c.}}}
\def\Re{\mathop{\mathrm{Re}}}
\def\Im{\mathop{\mathrm{Im}}}
\def\bra#1{\langle #1|}
\def\braa#1{\langle\langle #1|}
\def\ket#1{|\mbox{$#1$}\rangle}
\def\kett#1{|\mbox{$#1$}\rangle\rangle}
\def\bracket#1{\langle\mbox{$#1$}\rangle}
\def\bracketi#1#2{\langle\mbox{$#1$}|\mbox{$#2$}\rangle}
\def\bracketii#1#2#3{\langle\mbox{$#1$}|\mbox{$#2$}|\mbox{$#3$}\rangle}
\def\vct#1{\mathbf{#1}}
\def\fracpd#1#2{\frac{\partial#1}{\partial#2}}
\def\rank{\mathop{\mathrm{rank}}} 
\def\sub#1{_{\mbox{\scriptsize \rm #1}}}
\def\sur#1{^{\mbox{\scriptsize \rm #1}}}
\def\kagome{kagom\'{e} }
\def\kagometype{kagom\'{e}-type }
\def\Kagometype{Kagom\'{e}-type }

\title{
Plane-wave scattering by self-complementary metasurfaces\\
in terms of electromagnetic duality and Babinet's principle
}

\author{Yosuke Nakata}
\email{nakata@giga.kuee.kyoto-u.ac.jp}
\affiliation{Department of Electronic Science and Engineering, Kyoto University, Kyoto 615-8510, Japan}
\author{Yoshiro Urade}
\affiliation{Department of Electronic Science and Engineering, Kyoto University, Kyoto 615-8510, Japan}
\author{Toshihiro Nakanishi}
\affiliation{Department of Electronic Science and Engineering, Kyoto University, Kyoto 615-8510, Japan}
\author{Masao Kitano}
\email{kitano@kuee.kyoto-u.ac.jp}
\affiliation{Department of Electronic Science and Engineering, Kyoto University, Kyoto 615-8510, Japan}

\date{\today}

\begin{abstract}
We investigate theoretically electromagnetic plane-wave scattering by self-complementary metasurfaces.
By using Babinet's principle extended to metasurfaces with resistive elements, 
we show that the frequency-independent transmission and reflection are realized 
for normal incidence of a circularly polarized plane wave onto a self-complementary metasurface, 
even if there is diffraction.
Next, we consider two special classes of self-complementary metasurfaces.
We show that self-complementary metasurfaces with rotational symmetry can 
act as coherent perfect absorbers, 
and those with translational symmetry compatible with their self-complementarity
can split the incident power equally, even for oblique incidences.
\end{abstract}

\maketitle

\section{Introduction}

Metamaterials are artificially engineered materials
composed of lower-level components.\cite{Sihvola2007}
These components are called meta-atoms.
Various extraordinary electromagnetic 
properties of metamaterials have been predicted
and demonstrated, such as negative refractive index,
\cite{Veselago1968,Shelby2001a} artificial magnetism,\cite{Pendry1999}
super focusing,\cite{Pendry2000}  cloaking,\cite{Leonhardt2006,Pendry2006,Schurig2006}
and giant chirality.\cite{Rogacheva2006,Gansel2009} 

As in other fields of physics, such as crystallography 
and atomic or molecular spectroscopy, 
symmetry plays a fundamental role in metamaterials.
The symmetry of the shape or alignment of meta-atoms
affects the electromagnetic response of metamaterials.
A group-theoretical method of treating symmetry in metamaterials
has been developed and applied for designing and optimizing metamaterials.\cite{Wongkasem2006,Baena2007,Padilla2007, Reinke2011}
This method has also been utilized for designing two-dimensional metamaterials, called {\it metasurfaces}.\cite{Isik2009,Bingham2008} 
However, these studies dealt only with groups of isometries with a fixed point,
that is to say, {\it point groups}.

In addition to isometric symmetry of metamaterials, 
the theory of electromagnetism has another
symmetry with respect to the interchange of electric and magnetic fields.
This symmetry is called the {\it electromagnetic duality},
and can be generalized to a continuous symmetry with respect to 
internal rotations of electric and magnetic fields.
This continuous symmetry is directly related to a helicity conservation law.\cite{Calkin1965, Zwanziger1968, Deser1976, Drummond1999, Barnett2012, Cameron2012, Bliokh2013,Fernandez-Corbaton2013a}
We note that these symmetries had been gradually discovered
since the late 19th century.\cite{Heaviside1892, Larmor1897, Rainich1925}

The electromagnetic duality is closely related to Babinet's principle.\cite{Babinet1837}
Given a thin metallic metasurface, we can construct the complementary 
metasurface by using a structural inversion to interchange the holes and the metals. 
Babinet's principle relates the scattering fields of the 
complementary metasurfaces to those of the original one.
This principle is based on the fact that the structural inversion 
is consistent with electromagnetic duality.
A rigorous Babinet's principle for electromagnetic waves 
was simultaneously formulated by several 
groups.\cite{Booker1946,Copson1946,Mexner1946,Leontovich1946,*Landau1984, Kotani1947} 
It was extended to absorbing surfaces,\cite{Neugebauer1957}
impedance surfaces,\cite{Baum1974, *Baum1995} and surfaces with lumped elements.\cite{Moore1993}
It is important to note that the generalization for impedance surfaces was performed
by extending the structural inversion to the impedance one.
Recently, several complementary metasurfaces have been fabricated
and tested in the microwave,\cite{Falcone2004,Al-Naib2008}
terahertz,\cite{Chen2007} and near-infrared regions.\cite{Zentgraf2007}
Near-field images of complementary metasurfaces have 
been obtained in the terahertz range,\cite{Bitzer2011}
and switching of reflection has been realized by using a complementary metasurface 
with a twisted nematic cell in the near-infrared region.\cite{Lee2013}
Babinet's principle is also useful for designing
negative refractive index metamaterials.\cite{Zhang2013}

Generally, the structure of a metasurface is not invariant 
under impedance inversion.
If a metasurface is identical to its complement, 
it is called a {\it self-complementary} metasurface.
As an application, 
such self-complementary artificial surfaces
have been used for efficient polarizers.\cite{Beruete2007,Ortiz2013}
In the field of antenna design, it is known that
a self-complementary antenna has a constant input impedance.\cite{Mushiake1992, *Mushiake1996}
Therefore, self-complementary metasurfaces are expected to
exhibit a frequency-independent response.
It has been shown that an almost self-complementary spiral terahertz metasurface
has a constant response only in the high-frequency range.\cite{Singh2009}
There have been some efforts to achieve a frequency-independent 
response with self-complementary checkerboard metasurfaces,\cite{Compton1984,Takano2009}
but it is known empirically that such a metasurface does not exist.
Self-complementary metasurfaces 
have not been analyzed thoroughly enough; for example,  conditions for the frequency-independent 
response have not been discussed thoroughly, and an  elaborate theory is needed.
In this paper, we study electromagnetic scattering by self-complementary metasurfaces more rigorously and establish several useful theorems.
In particular, we focus on the incidences of circularly polarized plane waves 
onto self-complementary metasurfaces,
because circularly polarized light
matches with electromagnetic duality.

This article is organized as follows. In Sec.~\ref{sec:2},
we start by discussing the electromagnetic duality.
In Sec.~\ref{sec:3},
we review Babinet's principle for resistive metasurfaces,
and construct some relations between complex coefficients of transmission and reflection.
We analyze electromagnetic plane-wave scattering by self-complementary metasurfaces,
and derive their general properties in Sec.~\ref{sec:4}.
Numerical simulations are performed in order to confirm our theory in Sec.~\ref{sec:5}.
Finally, we summarize the conclusion in Sec.~\ref{sec:6}.

\section{\label{sec:2} Electromagnetic duality}
The electric and magnetic fields are represented by 
a polar vector field $\vct{E}$ and an axial vector field $\vct{H}$, respectively. 
Under spatial inversion, polar vectors are reversed in direction, while
axial vectors are invariant.
If we fix the orientation of the three-dimensional Euclid space $\mathbb{E}_3$,
axial vectors are represented by two polar vectors corresponding to the
two orientations of $\mathbb{E}_3$, 
respectively.\cite{Burke1985}
Two types of vectors are required in order not to assume
specific orientation of space $\mathbb{E}_3$.
An electromagnetic field  is represented by $(\vct{E},\vct{H})$.
The set of electromagnetic fields constitutes a vector space, namely, a 
direct sum of vector spaces,
with the scalar product defined by $s (\vct{E},\vct{H}):=(s\vct{E},s\vct{H})$ for a 
scalar $s$,
and the sum $(\vct{E}_1,\vct{H}_1)+(\vct{E}_2,\vct{H}_2):=(\vct{E}_1+\vct{E}_2,\vct{H}_1+\vct{H}_2)$.

Maxwell's theory of electromagnetism has an 
internal symmetry between electric and magnetic fields,
but the symmetry operation is not a simple exchange.
Maxwell's equations without sources and the vacuum constitutive relations ($\vct{D}=\varepsilon_0 \vct{E}$ and $\vct{B}=\mu_0 \vct{H}$) are invariant
under the following transformation:
\begin{align}
 \vct{E}&\rightarrow Z_0 \vct{H}, &\vct{H}&\rightarrow -\vct{E}/Z_0,  \label{eq:1}\\
\vct{D}&\rightarrow \vct{B}/Z_0, &\vct{B}&\rightarrow -Z_0\vct{D},  \label{eq:2}
\end{align} 
with an electric displacement $\vct{D}$ and magnetic flux density $\vct{B}$.
The permittivity,  permeability, and impedance of vacuum are represented by 
$\varepsilon_0$,  $\mu_0$, and $Z_0=\sqrt{\mu_0/\varepsilon_0}$, respectively.
This internal symmetry is called the ``electromagnetic duality.''
Note that
we need to fix an orientation of $\mathbb{E}_3$ to 
exchange polar and axial vectors by using Eqs.~(\ref{eq:1}) and (\ref{eq:2}).
This is similar to considering an imaginary number 
$\ii$ as an anti-clockwise rotation by $\pi/2$, 
which determines an orientation of the complex plane. It is also valid to use $\jj = -\ii$ as an anti-clockwise rotation
(this is the convention in engineering).
In the rest of this work, we use the right-handed system for internal transformations.

The electromagnetic duality extends to a continuous symmetry of electromagnetic fields.
The duality rotation by $\theta$ is defined by 
\begin{equation}
\mathcal{R}_\theta (\vct{E},\vct{H})
=(\vct{E} \cos\theta +Z_0\vct{H}\sin \theta, -\frac{\vct{E}}{Z_0} \sin\theta +\vct{H}\cos \theta).
  \label{eq:3}
\end{equation}
This transformation is considered to be a rotation with respect to the internal degree of freedom.
The transformation given by Eqs.~(\ref{eq:1}) and (\ref{eq:2}) corresponds to the duality rotation by $\theta=\pi/2$.
The duality rotation mixes the two linear polarized plane waves.
Here, we use tildes to represent the complex amplitudes 
for sinusoidally oscillating fields with angular frequency $\omega$.
For example, a sinusoidally oscillating real-valued scalar field $F$ is represented by 
$F=\tilde{F}\ee^{-\ii\omega t} + \tilde{F}^*\ee^{\ii\omega t}$, where $\tilde{F}$ is the complex 
amplitude and $\tilde{F}^*$ is its complex conjugate.
With this notation,
we have $ \mathcal{R}_\theta (\tilde{\vct{E}}_{\mathrm{LCP}},\tilde{\vct{H}}_{\mathrm{LCP}})  = \ee^{-\ii \theta} (\tilde{\vct{E}}_{\mathrm{LCP}},\tilde{\vct{H}}_{\mathrm{LCP}})$ for a left circularly polarized wave $(\tilde{\vct{E}}_{\mathrm{LCP}},\tilde{\vct{H}}_{\mathrm{LCP}})$ from the point of view of the receiver.
For a right circularly polarized wave $(\tilde{\vct{E}}_{\mathrm{RCP}},\tilde{\vct{H}}_{\mathrm{RCP}})$,
$ \mathcal{R}_\theta (\tilde{\vct{E}}_{\mathrm{RCP}},\tilde{\vct{H}}_{\mathrm{RCP}}) = \ee^{\ii \theta}(\tilde{\vct{E}}_{\mathrm{RCP}},\tilde{\vct{H}}_{\mathrm{RCP}})$
is satisfied.
Therefore, circularly polarized plane waves are eigenstates of $\mathcal{R}_\theta$.

\section{\label{sec:3} Babinet's principle for metasurfaces with resistive elements}
In this section, we derive Babinet's principle for metasurfaces with resistive elements.
In deriving Babinet's principle,
we consider two scattering problems.
We assume that a metasurface is placed in a vacuum for each problem.
\begin{figure}[!tbp]
\includegraphics[width=85mm]{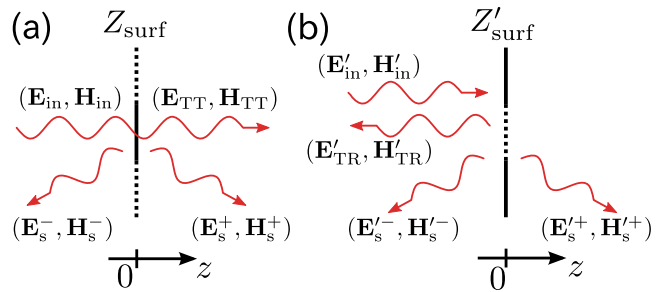}
\caption{\label{fig:two_problems} (Color online) Two problems for Babinet's principle}
\end{figure}
In the first problem [problem~(a)], 
an incident electromagnetic wave $(\vct{E}\sub{in},\vct{H}\sub{in})$
is scattered by the metasurface with 
a surface impedance $Z\sub{surf}$
on the surface $z=0$ [see Fig.~\ref{fig:two_problems}(a)].
Note that $Z\sub{surf}$ is a function of $(x,y)$, but we omit the parameters for simplicity.
The incident wave radiates from the sources in $z\leq 0$.
If there was no metasurface,
the source would produce $(\vct{E}\sub{in},\vct{H}\sub{in})$ in $z\leq 0$ and $(\vct{E}\sub{TT},\vct{H}\sub{TT})$ in $z \geq 0$.
Here, $\mathrm{TT}$ represents a totally transmitted wave.
The incident wave is not restricted to plane waves and 
can even include near-field components.
If there is a metasurface,
surface currents and charges are induced by the incident wave.
They radiate the following scattered fields: $(\vct{E}^{-}\sub{s}, \vct{H}^{-}\sub{s})$ in $z\leq0$ 
and $(\vct{E}^{+}\sub{s}, \vct{H}^{+}\sub{s})$ in $z\geq0$.

Next, we set up the second problem [problem~(b)].
In this problem, an incident wave $(\vct{E}\sub{in}',\vct{H}\sub{in}')$
from sources in $z\leq0$ enters the metasurface at $z=0$ with 
a surface impedance $Z\sub{surf}'$ varying on $z=0$ [see Fig.~\ref{fig:two_problems}(b)].
Here, $(\vct{E}\sub{in}',\vct{H}\sub{in}')$ is defined 
in $z\leq0$. If a perfect electric conductor (PEC) sheet is placed at $z=0$,
the incident wave is totally reflected. This totally reflected wave in $z\leq 0$ is represented by $(\vct{E}\sub{TR}',\vct{H}\sub{TR}')$.
The effect of the metasurface that differs from the PEC sheet
emerges as the remaining fields $(\vct{E}'^\pm\sub{s}, \vct{H}'^\pm\sub{s})$,
where $-$ and $+$ represent the fields in
$z\leq 0$ and $z\geq 0$, respectively.

In general, these two problems are completely distinct. 
If we assume a specific condition for the surface impedances,
the scattering fields of the two problems are related as described in the following theorem.
\begin{theo}\label{theo:1}
If $Z\sub{surf}$ and $Z\sub{surf}'$
satisfy $Z\sub{surf}\, Z'\sub{surf}=\left(Z_0/2\right)^2$ at any point with $z=0$, 
the scattering fields of problem~(b) are given by
 $(\vct{E}'^\pm\sub{s}, \vct{H}'^\pm\sub{s})=\mathcal{R}_{\pm\pi/2} (\vct{E}^\pm\sub{s},\vct{H}^\pm\sub{s})=\pm(Z_0\vct{H}^\pm\sub{s}, -\vct{E}^\pm\sub{s}/Z_0)$
for the incident wave $(\vct{E}\sub{in}',\vct{H}\sub{in}')=
\mathcal{R}_{-\pi/2}(\vct{E}\sub{in},\vct{H}\sub{in})=
(-Z_0\vct{H}\sub{in},\vct{E}\sub{in}/Z_0)$ using the solution of problem~(a).
\end{theo}
\begin{proof}
Here, we define a unit vector $\vct{e}_z$ parallel to the $z$ axis, 
and the projection operator $\mathcal{P}=-\vct{e}_z\times\vct{e}_z\times$,
which eliminates $z$-components of vectors.
First, we consider problem~(a).
The scattered fields $(\vct{E}^\pm\sub{s}, \vct{H}^\pm\sub{s})$ are symmetric with respect to  $z=0$.
Then, $\mathcal{P} \vct{E}\sub{s}^+ =\mathcal{P} \vct{E}\sub{s}^-$ and $\mathcal{P} \vct{H}\sub{s}^+ =-\mathcal{P} \vct{H}\sub{s}^-$ are satisfied on $z=0$.
The electric boundary condition $\mathcal{P}(\vct{E}\sub{TT}+\vct{E}^{+}\sub{s})=\mathcal{P}(\vct{E}\sub{in}+\vct{E}^{-}\sub{s})$ on $z=0$ 
is automatically satisfied. 
Another boundary condition on $z=0$ is given by 
$\mathcal{P}(\vct{E}\sub{in}+\vct{E}^-\sub{s})=Z\sub{surf}\,{\vct{e}_z}\times(\vct{H}^+\sub{s}-\vct{H}^-\sub{s})$.
With $\mathcal{P}\vct{H}^{+}\sub{s}=-\mathcal{P}\vct{H}^{-}\sub{s}$, 
we obtain the following equation for $z=0$:
\begin{equation}
 \mathcal{P}(\vct{E}\sub{in}+\vct{E}^-\sub{s})=2Z\sub{surf}\,{\vct{e}_z}\times\vct{H}^+\sub{s}.  \label{eq:4}
\end{equation}

In problem~(b), we show that the fields $(\vct{E}'^\pm\sub{s},\vct{H}'^\pm\sub{s})$
defined by $\mathcal{R}_{\pm\pi/2} (\vct{E}^\pm\sub{s},\vct{H}^\pm\sub{s})$ 
satisfy all boundary conditions for the incident wave $(\vct{E}\sub{in}',\vct{H}\sub{in}')=
\mathcal{R}_{-\pi/2}(\vct{E}\sub{in},\vct{H}\sub{in})$.
The fields $(\vct{E}'^\pm\sub{s},\vct{H}'^\pm\sub{s})=\pm(Z_0\vct{H}^\pm\sub{s}, -\vct{E}^\pm\sub{s}/Z_0)$ are also symmetric with respect to $z=0$.
From $\mathcal{P} \vct{E}'\sub{in}=-\mathcal{P}\vct{E}'\sub{TR}$ and
$\mathcal{P} \vct{E}'^-\sub{s} = \mathcal{P}\vct{E}'^+\sub{s}$ on $z=0$,
the electric boundary condition $\mathcal{P}(\vct{E}'\sub{in}+\vct{E}'\sub{TR}+\vct{E}'^{-}\sub{s})=\mathcal{P}\vct{E}'^+\sub{s}$ is satisfied.
Additionally, the following boundary condition should be satisfied on $z=0$:
\begin{align}
 \mathcal{P}\vct{E}'^+\sub{s}&=
Z\sub{surf}'\,{\vct{e}_z}\times(\vct{H}'^+\sub{s}-\vct{H}'^-\sub{s}-\vct{H}'\sub{in}-\vct{H}'\sub{TR})
\nonumber \\
&=-2Z\sub{surf}'\,{\vct{e}_z}\times(\vct{H}\sub{in}'+\vct{H}'^-\sub{s}),
  \label{eq:5}
\end{align}
where we use $\mathcal{P} \vct{H}'^+\sub{s}=-\mathcal{P}  \vct{H}'^-\sub{s}$
and
$\mathcal{P}\vct{H}'\sub{in}=\mathcal{P}\vct{H}'\sub{TR}$ on $z=0$
(the derivation of $\mathcal{P}\vct{H}'\sub{in}=\mathcal{P}\vct{H}'\sub{TR}$ 
is shown in Appendix~\ref{append1}).
Operating with $\vct{e}_z\times$ on Eq.~(\ref{eq:5}) and comparing with Eq.~(\ref{eq:4}),
we have $Z\sub{surf}\, Z\sub{surf}'=\left(Z_0/2\right)^2$.
Thus all boundary conditions are satisfied for problem~(b) 
with $Z\sub{surf}'={Z_0}^2/(4Z\sub{surf})$.
\end{proof}
For the case of $Z\sub{surf}=\infty$ (hole), the complementary surface is 
PEC with $Z\sub{surf}'=0$, and vice versa.
Therefore, the above theorem includes the standard Babinet's principle.
The extensions for tensor impedances\cite{Baum1974} and lumped elements\cite{Moore1993} have also been investigated.

Next, we discuss the relationship of the transmission and reflection 
coefficients in problems~(a) and (b).
From here on, we assume that all fields oscillate sinusoidally with angular frequency $\omega$
and are represented by complex amplitudes.
We consider a periodic metasurface with lattice vectors $\vct{a}_1$ and
$\vct{a}_2$. Physically, a metasurface without periodicity can be regarded as 
$|\vct{a}_1|,\ |\vct{a}_2| \rightarrow \infty$. This corresponds to the
transition from box quantization to free space quantization in
quantum mechanics.
The reciprocal vectors are represented by $\vct{b}_1$ and $\vct{b}_2$ satisfying
$\vct{a}_i\cdot\vct{b}_j=2\pi \delta_{ij}$ ($\delta_{ij}$ is the Kronecker delta).
Additionally, we assume that the incident wave is a plane wave 
$(\tilde{\vct{E}}\sub{in}, \tilde{\vct{H}}\sub{in})=
(\check{\vct{E}}_0\ee^{\ii \vct{k}_0 \cdot \vct{x}},
\check{\vct{H}}_0\ee^{\ii \vct{k}_0 \cdot \vct{x}} )$
with $\check{{\vct{H}}}_0={Z_0}^{-1} \vct{k}_0 \times \check{\vct{E}}_0/|\vct{k}_0|$
for a wave vector $\vct{k}_0$. Here we use the checkmark symbol in order to express
complex amplitudes for a plane wave with a definite wave vector.

In problem~(a), the scattered wave 
$(\tilde{\vct{E}}^+\sub{s}, \tilde{\vct{H}}^+\sub{s})$ on $z=0$ 
has Fourier components with
the in-plane wave vector $\underline{\vct{k}}_{pq}:= p \vct{b}_1+q\vct{b}_2+\mathcal{P} \vct{k}_0$ for $(p,q)\in \mathbb{Z}^2$. 
In this paper, we focus on the 0th-order modes with $(p,q)=(0,0)$
in order to simplify the notation.
The general case is summarized in Appendix~\ref{append2}.
We decompose the 0th-order complex fields of problem~(a) in $z\geq 0$
as $\sum_{\alpha=1,2} t_{\alpha}
(\tilde{\vct{E}}^+_{\alpha}, \tilde{\vct{H}}^+_{\alpha})$
with complex transmission coefficients $t_{\alpha}$,
where we define $(\tilde{\vct{E}}^+_{1},\tilde{\vct{H}}^+_{1}):=
(\tilde{\vct{E}}_{\mathrm{TT}}, \tilde{\vct{H}}_{\mathrm{TT}})$, 
and its perpendicular polarization state $(\tilde{\vct{E}}^+_{2},\tilde{\vct{H}}^+_{2})$.
The mode $(\tilde{\vct{E}}^+_{2},\tilde{\vct{H}}^+_{2})$ is normalized to 
carry the same power flow of $(\tilde{\vct{E}}^+_{1},\tilde{\vct{H}}^+_{1})$.
We also define $(\tilde{\vct{E}}^-_{\alpha}, \tilde{\vct{H}}^-_{\alpha})$ as 
the mirror symmetric field of $(\tilde{\vct{E}}^+_{\alpha}, \tilde{\vct{H}}^+_{\alpha})$ with respect to $z=0$. In $z\leq 0$, the 0th-order field is represented by 
$(\tilde{\vct{E}}\sub{in}, \tilde{\vct{H}}\sub{in})+\sum_{\alpha=1,2} r_{\alpha}
(\tilde{\vct{E}}^-_{\alpha}, \tilde{\vct{H}}^-_{\alpha})$ with complex reflection coefficients $r_{\alpha}$.
For problem~(b), 
we define
 $(\tilde{\vct{E}}'^\pm_{\alpha}, \tilde{\vct{H}}'^\pm_{\alpha}):=
\mathcal{R}_{\mp \pi/2} (\tilde{\vct{E}}^\pm_{\alpha}, \tilde{\vct{H}}^\pm_{\alpha})$.
The 0th-order fields are written as
 $\sum_{\alpha=1,2} t'_{\alpha}(\tilde{\vct{E}}'^+_{\alpha}, \tilde{\vct{H}}'^+_{\alpha})$
in $z \geq 0$, and 
$(\tilde{\vct{E}}'\sub{in}, \tilde{\vct{H}}'\sub{in})+\sum_{\alpha=1,2}r'_{\alpha}(\tilde{\vct{E}}'^-_{\alpha}, \tilde{\vct{H}}'^-_{\alpha})$
in $z\leq 0$.

Now we formulate Babinet's principle for complex coefficients as follows:
\begin{theo}
\label{theo:2}
The coefficients of the two problems are related as
$t_{1}+t_{1}'=1$,
$r_{1}+r'_{1}=-1$,
and 
$t_{2}+t'_{2}=0$,
$r_{2}+r'_{2}=0$.
\end{theo}
\begin{proof}
For problem~(a), the 0th-order component of the scattered field in $z\geq 0$ is given by
\begin{equation}
(\tilde{\vct{E}}^+\sub{s0},\tilde{\vct{H}}^+\sub{s0})= - (\tilde{\vct{E}}^+_{1}, \tilde{\vct{H}}^+_{1}) + \sum_{\alpha=1,2}
t_{\alpha} (\tilde{\vct{E}}^+_{\alpha}, \tilde{\vct{H}}^+_{\alpha}).  \label{eq:6}
\end{equation}
In problem~(b), the 0th-order component of $(\tilde{\vct{E}}'^+\sub{s},\tilde{\vct{H}}'^+\sub{s})$
is 
\begin{equation}
(\tilde{\vct{E}}'^+\sub{s0},\tilde{\vct{H}}'^+\sub{s0})= \sum_{\alpha=1,2} t'_{\alpha}(\tilde{\vct{E}}'^+_{\alpha}, \tilde{\vct{H}}'^+_{\alpha}).  \label{eq:7}
\end{equation}
Applying $\mathcal{R}_{-\pi/2}$ to Eq.~(\ref{eq:7}),
we have 
\begin{equation}
 (\tilde{\vct{E}}^+\sub{s0},\tilde{\vct{H}}^+\sub{s0})= -\sum_{\alpha=1,2} t'_{\alpha}(\tilde{\vct{E}}^+_{\alpha}, \tilde{\vct{H}}^+_{\alpha}),  \label{eq:8}
\end{equation}
where we use 
$(\vct{E}'^+\sub{s0}, \vct{H}'^+\sub{s0})=\mathcal{R}_{\pi/2} (\vct{E}^+\sub{s0},\vct{H}^+\sub{s0})$ and 
$(\vct{E}'^+_\alpha, \vct{H}'^+_\alpha)=\mathcal{R}_{-\pi/2} (\vct{E}^+_\alpha, \vct{H}^+_\alpha)$.
Comparing Eq.~(\ref{eq:8}) with Eq.~(\ref{eq:6}),
we obtain $t_{1}+t_{1}'=1$ and $t_{2}+t_{2}'=0$.
The remaining equations are derived from a similar discussion for $z\leq 0$.
\end{proof}

\section{Self-complementary metasurfaces \label{sec:4}}
For a metasurface with a surface impedance $Z\sub{surf}$,
we can create the complementary metasurface with $Z\sub{surf}'={Z_0}^2/(4Z\sub{surf})$.
This operation is called an {\it impedance inversion about $Z_0/2$}.
Two metasurfaces are congruent if
one can be transformed into the other by a combination of translations, 
rotations and reflections.
When a metasurface is congruent to its complementary one,
we say that it is {\it self-complementary}.
We emphasize that the self-complementarity is
not the same as the point-group symmetry.
Several examples of self-complementary metasurfaces are
shown in Fig.~\ref{fig:classification}.

\begin{figure*}[!htbp]
\includegraphics[width=140mm]{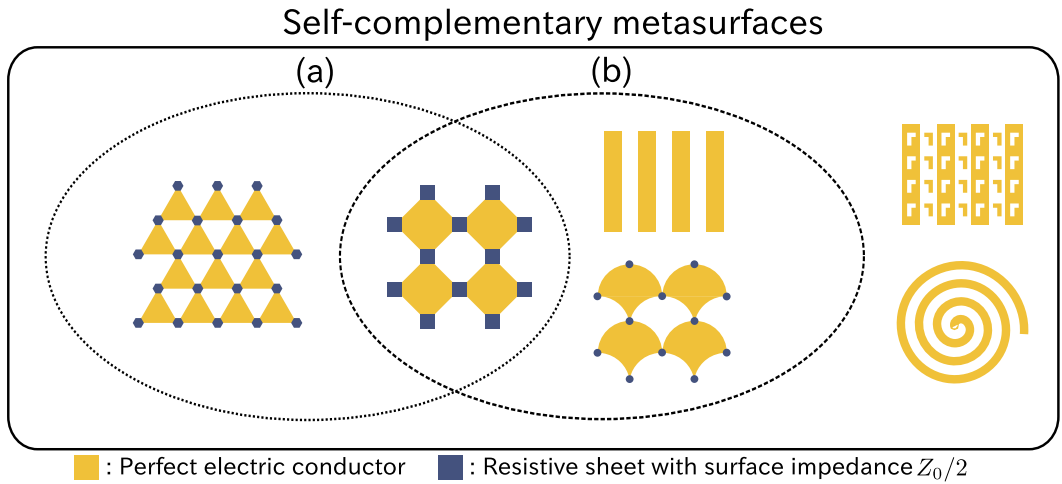}
\caption{\label{fig:classification} (Color online) Examples of self-complementary metasurfaces.
The class of self-complementary metasurfaces includes
 two specific subclasses with (a) $n$-fold rotational symmetry ($n\geq 3$)
 and (b) translational self-complementarity.}
\end{figure*}

For a left circularly polarized incident wave,
we define $t\sub{LL}:= t_{1}$, $t\sub{RL}:=t_{2}$, and
$r\sub{RL}:=r_{1}$, $r\sub{LL}:=r_{2}$.
We also use $t\sub{RR}:=t_{1}$,  $t\sub{LR}:=t_{2}$, and
$r\sub{LR}:=r_{1}$, $r\sub{RR}:=r_{2}$
for a right circularly polarized incident wave.
From reciprocity and the mirror symmetry of a metasurface with respect to $z=0$,
the following theorem is derived.
\begin{theo}
\label{theo:3}
In the case of normal incidence of a circularly 
polarized plane wave onto a metasurface, 
$t\sub{RR}=t\sub{LL}$ and 
$r\sub{LR}=r\sub{RL}$ are satisfied.
\end{theo}
\begin{proof}
We consider two situations.
In the first, the incident wave 
is a left circularly polarized wave 
$(\tilde{\vct{E}}\sub{in},\tilde{\vct{H}}\sub{in})=(\check{E}_0 \vct{e}_+ ,-\ii \check{H}_0 \vct{e}_+  )\ee^{\ii k_0 z }$ from $z\leq 0$, where $\check{H}_0= \check{E}_0/Z_0$
and $\vct{e}_\pm:=(\vct{e}_{x} \pm \ii \vct{e}_{y})/\sqrt{2}$ with 
unit vectors $\vct{e}_x$ and $\vct{e}_y$ along $x$ and $y$ axes.
The total field is represented by $(\tilde{\vct{E}}\sub{f}, \tilde{\vct{H}}\sub{f})$.
In the second situation, an incident wave from $z\geq0$ is 
 $(\tilde{\vct{E}}\sub{in},\tilde{\vct{H}}\sub{in})=(\check{E}_0 \vct{e}_-  ,\ii \check{H}_0 \vct{e}_- ) \ee^{-\ii k_0 z }$, and the total field is denoted by $(\tilde{\vct{E}}\sub{b}, \tilde{\vct{H}}\sub{b})$.
If we perform the coordinate transformation $z \rightarrow -z$, 
the second situation can be transformed to the scattering problem for a right circularly 
polarized incident wave from $z\leq 0$,
because of the symmetry between the front and back of the metasurface.
Then, $t_1$ of the second situation is equivalent to
$t\sub{RR}$.


We represent the unit cell on $z = 0$ by $U$, and consider
$V = U \times [h/2,h/2]$ with $h>0$.
For the normal incidence, we can impose periodic boundary conditions 
on two pairs of opposite faces of $\partial U \times [h/2,h/2]$.
From the Lorentz reciprocity theorem\cite{Collin1991}
\begin{equation} 
 \int_{\partial V} (\tilde{\vct{E}}\sub{f}\times \tilde{\vct{H}}\sub{b}
-\tilde{\vct{E}}\sub{b}\times \tilde{\vct{H}}\sub{f})\cdot \dd \vct{S} =0  \label{eq:9}
\end{equation}
and $\vct{e}_+\times\vct{e}_-=-\ii \vct{e}\sub{z}$,
we obtain $t\sub{RR}=t\sub{LL}$.
Because electric fields are continuous on $z=0$,
$1+r\sub{LR}=t\sub{RR}$ and $1+r\sub{RL}=t\sub{LL}$
are satisfied. Then, $r\sub{LR}=r\sub{RL}$ is proved.
\end{proof}

By using Theorems~\ref{theo:2} and \ref{theo:3}, 
we can arrive at the following theorem.
\begin{theo}
\label{theo:4}
In the case of normal incidence of a circularly 
polarized plane wave onto a self-complementary metasurface, 
$t\sub{RR}=t\sub{LL}=1/2$ and
$r\sub{LR}=r\sub{RL}=-1/2$ are satisfied.
\end{theo}
\begin{proof}

This situation is regarded as problem~(a) shown in Fig.~\ref{fig:two_problems}.
Because the metasurface is self-complementary, its complement 
returns to the original metasurface by the finite numbers of reflections.
The product of these operations is denoted by $\mathcal{X}$.
Problem~(b) related to problem~(a) through Theorem~\ref{theo:1} is considered.
Applying $\mathcal{X}$ to all fields and structures of problem~(b),
we have problem~(c).
Now, we consider the two cases where even and odd numbers of the reflections 
are involved in $\mathcal{X}$.
In the even case,
$(\tilde{\vct{E}}\sub{in}, \tilde{\vct{H}}\sub{in})$ is 
an eigenmode for $\mathcal{X}$.
Therefore, problem~(c) is identical to problem~(a)
except for the total phase, and $t'_{1}=t_{1}$ is satisfied,
where the transmission coefficient $t'_1$ of problem~(b) is defined in Sec.~\ref{sec:3}.
In the case of odd reflections,
the polarization is changed by $\mathcal{X}$ (for example, from LCP to RCP),
but Theorem~\ref{theo:3} assures $t'_{1}=t_{1}$.
Finally, we obtain $t'_{1}=t_{1}=1/2$ from Theorem~\ref{theo:2} for both cases.
\end{proof}
We note that the frequency-independent transmission of 
self-complementary metasurfaces is
valid in the high-frequency range with diffraction.
In the following, we consider subclasses of self-complementary metasurfaces
shown in Fig.~\ref{fig:classification}.

If a metasurface has rotational symmetry
in addition to self-complementarity [see Fig.~\ref{fig:classification}(a)], 
we have the following theorem.
\begin{theo}
\label{theo:5}
For normal incidence of a plane wave with an arbitrary polarization onto
a self-complementary metasurface with 
$n$-fold rotational symmetry $(n\geq 3)$, 
$t_{1}=1/2$, $r_{1}=-1/2$ and $t_{2}=0$, $r_{2}=0$ are satisfied.
Half the incident power is absorbed by the metasurface
in the frequency range without diffraction.
\end{theo}

\begin{proof}
We consider two incident waves in $z\leq 0$:
$(\tilde{\vct{E}}\sur{L}\sub{in},\tilde{\vct{H}}\sub{in}\sur{L})=(\check{E}_0 \vct{e}_+ ,-\ii \check{H}_0 \vct{e}_+  )\ee^{\ii k_0 z }$ and 
$(\tilde{\vct{E}}\sub{in}\sur{R},\tilde{\vct{H}}\sub{in}\sur{R})=(\check{E}_0 \vct{e}_- ,\ii \check{H}_0 \vct{e}_-  )\ee^{\ii k_0 z }$.
Here, we define $(\tilde{\vct{E}}_{\alpha}^{+, \beta},
\tilde{\vct{H}}_{\alpha}^{+,\beta})$ as $(\tilde{\vct{E}}_{\alpha}\sur{+},\tilde{\vct{H}}_{\alpha}\sur{+})$ for an incident wave with polarization $\beta$.
We adjust the phase of $(\tilde{\vct{E}}_{2}\sur{+},\tilde{\vct{H}}_{2}\sur{+})$ 
so as to satisfy
$(\tilde{\vct{E}}_{2}\sur{+,R},\tilde{\vct{H}}_{2}\sur{+,R})=(\tilde{\vct{E}}_{1}\sur{+,L},\tilde{\vct{H}}_{1}\sur{+,L})$ and 
$(\tilde{\vct{E}}_{2}\sur{+,L},\tilde{\vct{H}}_{2}\sur{+,L})=(\tilde{\vct{E}}_{1}\sur{+,R},\tilde{\vct{H}}_{1}\sur{+,R})$ on $z=0$, for each incident wave.
In this situation, we can define a complex transmittance matrix
 \begin{equation}
  \tau = 
\begin{bmatrix}
 t\sub{LL}& t\sub{LR}\\
 t\sub{RL}& t\sub{RR}
\end{bmatrix}.  \label{eq:10}
 \end{equation}
For a circularly polarized basis, a rotation by $2\pi/n$ about $z$ axis is represented by
 \begin{equation}
P:=\begin{bmatrix}
\ee^{-\ii \frac{2\pi}{n}}&0\\
0&\ee^{\ii  \frac{2\pi}{n}}
\end{bmatrix}.  \label{eq:11}
 \end{equation}
Because of $n$-fold symmetry, $P^{-1} \tau P = \tau$ is satisfied, and
then $t\sub{RL}=t\sub{LR}=0$.
Therefore, we obtain 
\begin{equation}
 \tau=\begin{bmatrix}
\frac{1}{2}&0\\
0&\frac{1}{2}
\end{bmatrix},  \label{eq:12}
\end{equation}
from Theorem~\ref{theo:4}.
Because $\tau$ is proportional to the identity matrix, 
$t_{1}=1/2$, $r_{1}=-1/2$
and $t_{2}=0$, $r_{2}=0$ are satisfied for an incident plane wave with an arbitrary polarization. 

In the frequency range without diffraction, 
the Fourier components of $(\tilde{\vct{E}}^{\pm}\sub{s}, \tilde{\vct{H}}^{\pm}\sub{s})$ with $(p,q)\ne 0$  
are evanescent waves.
For evanescent waves, the real part of the $z$-component 
of Poyinting vectors is zero; therefore, they do not carry energy out of $z=0$.
The remaining power $A=1-|t_{1}|^2-|r_{1}|^2=1/2$ is absorbed in the metasurface.
\end{proof}

From Theorem~\ref{theo:5}, we find that the metasurface can absorb total incident energy as follows:
\begin{theo}
\label{theo:6}
If we excite a self-complementary metasurface with 
$n$-fold rotational symmetry $(n\geq 3)$ 
by two in-phase plane waves 
$(\check{\vct{E}}_0, \check{\vct{H}}_0)\ee^{\ii k_0 z}$ from $z\leq 0$
and $(\check{\vct{E}}_0, -\check{\vct{H}}_0)\ee^{-\ii k_0 z}$ from $z\geq 0$
with an arbitrary polarization $(\check{\vct{E}}_0, \check{\vct{H}}_0)$,
the incident power is perfectly absorbed
in the frequency range without diffraction.
\end{theo}
\begin{proof}
In the case of one excitation $(\check{\vct{E}}_0, \check{\vct{H}}_0)\ee^{\ii k_0 z}$ from $z\leq 0$, half of the power is absorbed in the frequency range without diffraction.
If we excite from both sides in phase, the electric field is doubled, 
and then absorption is quadrupled. Therefore, all of the incident power is absorbed.
\end{proof}
If we excite the above self-complementary metasurface 
by two antiphase plane waves, there is no absorption.
This is because boundary conditions at $z=0$ are already satisfied
without induced currents and charges. The perfect absorption is only realized when two beams have the correct relative phase and amplitude.
This function is referred to as coherent perfect absorption.\cite{Chong2010,Wan2011}
We note that self-complementarity is not a necessary condition for 
the frequency-independent response described in Theorem~\ref{theo:7} because a
 similar frequency-independent 
response can be seen in other systems, 
such as percolated metallic films\cite{Yagil1987,Yagil1988,Gadenne1988,Gadenne1989,Beghdadi1989,Davis1991,Sarychev1995} and two identical lamellar gratings.\cite{Botten1997}

There is another interesting class of 
self-complementary metasurfaces.
If a metasurface returns to the original one by just a translation after the impedance inversion about $Z_0/2$,
we say that it has {\it translational self-complementarity} [see Fig.~\ref{fig:classification}(b)].
This subclass of self-complementary metasurfaces has the following property.
\begin{theo}
\label{theo:7}
In the case of an oblique incidence of a circularly polarized plane wave onto
a metasurface with translational self-complementarity, 
$t\sub{RR}=t\sub{LL}=1/2$ and
$r\sub{LR}=r\sub{RL}=-1/2$ are satisfied.
\end{theo}
\begin{proof}

We regard this situation as problem~(a) shown in Fig.~\ref{fig:two_problems}.
The metasurface returns to the original position
by a translation $\mathcal{T}$ together with the impedance inversion.
Problem~(b) can be related to problem~(a) through Theorem~\ref{theo:1}.
We introduce problem~(c) in which 
the incident wave and the metasurface of problem~(b) 
are translated by $\mathcal{T}$.
From the definition of $\mathcal{T}$, the metasurface of 
problem~(c) is the same as that in problem~(a).
The incident field of problem~(c) is written as
$\mathcal{T} \mathcal{R}_{-\pi/2}(\tilde{\vct{E}}\sub{in}, \tilde{\vct{H}}\sub{in})$.
Because $(\tilde{\vct{E}}\sub{in}, \tilde{\vct{H}}\sub{in})$ is 
an eigenmode for $\mathcal{T} \mathcal{R}_{-\pi/2}$,
the incident wave of problem~(c) is identical to 
that of problem~(a) except for the total phase.
In this way, $t'_1=t_1$ and $r'_1=r_1$ are confirmed.
From Theorem~\ref{theo:2}, we have
$t'_{1}=t_{1}=1/2$ and $r'_{1}=r_{1}=-1/2$.
\end{proof}
This theorem shows that self-complementary metasurfaces can be used as
beam splitters. 
The extension for general diffraction orders is discussed in 
Appendix~\ref{append3}.

\section{\label{sec:5}Examples: checkerboard metasurfaces}
In this section, we apply the current theory for
checkerboard metasurfaces\cite{Compton1984,Takano2009, Edmunds2010, Kempa2010, Ramakrishna2011} and confirm its validity by simulations.
First, we consider an ideal checkerboard metasurface shown 
in Fig.~\ref{fig:checker_board}(a).
It is expected that the ideal checkerboard metasurface should
exhibit a frequency-independent response
because of its self-complementarity. 
However, it is known that the ideal checkerboard 
metasurface cannot be realized.
This is explained as follows.\cite{Takano2009} The electromagnetic response 
of the checkerboard metasurface drastically changes
depending on whether the square metals are connected or not.
The transmittance and reflectance do not converge
when the structure approaches 
the ideal checkerboard metasurface.
Furthermore, it has also been reported that there is 
an instability in numerical calculations for the ideal checkerboard metasurface,
and the checkerboard metasurfaces exhibit percolation effects near the ideal checkerboard metasurface.\cite{Kempa2010}
\begin{figure}[!bt]
\includegraphics[width=80mm]{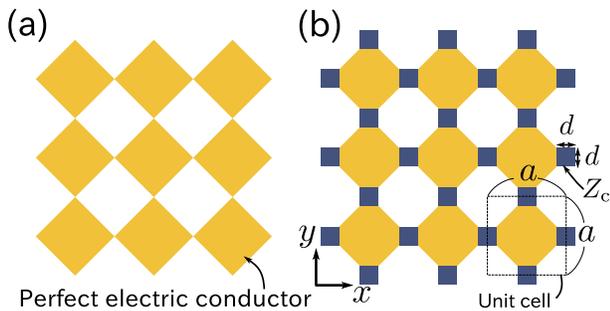}
\caption{\label{fig:checker_board}
(Color online) (a) Ideal checkerboard metasurface (b) Resistive checkerboard metasurface}
\end{figure}

By using our theory, we can give another explanation 
without relying on asymptotic behaviors.
From Theorem~\ref{theo:5}, 
the power transmission $T=|t_1|^2$ and reflection $R=|r_1|^2$
should satisfy $T=R=1/4$ for
the ideal checkerboard metasurface with 4-fold rotational symmetry.
However, energy conservation means $T+R=1$
in the frequency range without diffraction,
because there is no absorption
in the perfect checkerboard metasurface.
This contradiction implies that the ideal checkerboard metasurface is unphysical.

The above explanation gives us another insight:
we may realize the frequency-independent response of a checkerboard metasurface
if resistive elements are introduced.
We replace the singular contacts with tiny resistive sheets 
with a surface impedance $Z\sub{c}$
and obtain {\it a resistive checkerboard metasurface} 
shown in Fig.~\ref{fig:checker_board}(b).
When $Z\sub{c}=Z_0/2$ is satisfied, the resistive checkerboard metasurface
is self-complementary and is expected to exhibit a frequency-independent response. 

\begin{figure*}[!htb]
\includegraphics[width=170mm]{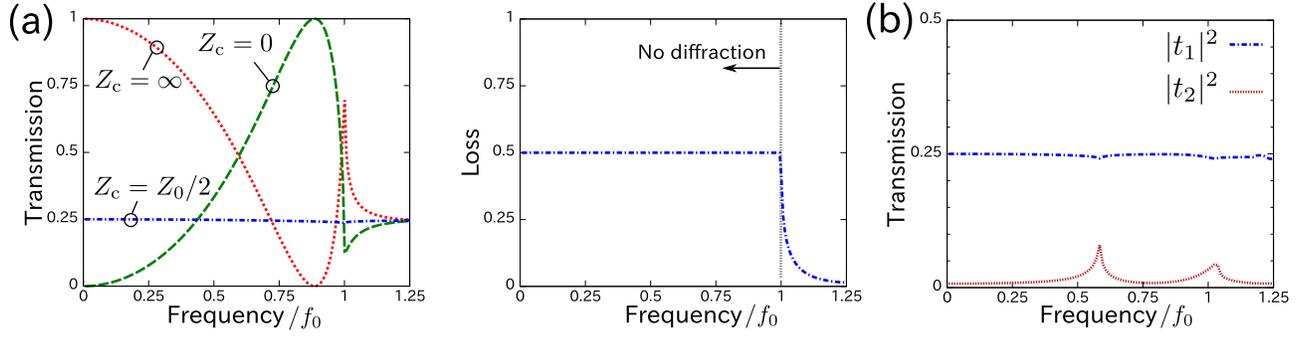}
\caption{\label{fig:simulation} 
(Color online) (a) Power transmission and loss spectra for normally incident $x$-polarized plane waves entering into resistive checkerboard metasurfaces with $d/a=0.2$. 
The left graph shows power transmission spectra for resistive checkerboard metasurfaces with $Z\sub{c}=0,\, Z_0/2$, and $\infty$.
The right graph shows the loss spectrum calculated for the self-complementary resistive checkerboard metasurface with $Z\sub{c}=Z_0/2$.
(b) The calculated spectra of $|t_{1}|^2$ and  $|t_{2}|^2$
 for an oblique incident circularly polarized plane wave entering into the resistive checkerboard metasurfaces with $Z\sub{c}=Z_0/2$ and $d/a=0.2$.
The incident wave has a wave vector $\vct{k}_0= (k_0 \sin\theta,0, k_0 \cos \theta)$ with $\theta=\pi/4$.
In these simulations, the frequencies are normalized by $f_0=c_0/a$.
}
\end{figure*}

For confirmation of our theory, we calculate the electromagnetic 
response of resistive checkerboard metasurfaces on $z=0$
by a commercial finite-element method solver (\textsc{Ansoft HFSS}).
In the simulation, normal incident $x$-polarized plane wave
is injected onto 
resistive checkerboard metasurfaces with $d/a=0.2$,
where $a$ and $d$ are the side length of the square unit cell and 
that of impedance sheet, respectively.
By imposing periodic boundaries on four sides, 
the electromagnetic fields in the unit cell
are calculated for $Z\sub{c}=0,\, Z_0/2$, and $\infty$.
We take into account diffracted modes with
$\{(p,q,\alpha)|-1\leq p \leq 1,\ -1\leq q \leq 1,\ \alpha=1,2\}$  (18 modes),
where $(p,q)$ and $\alpha$ are defined in Sec.~\ref{sec:3} and Appendix~\ref{append2}.

The left panel of Fig.~\ref{fig:simulation}(a) displays the spectra of 
power transmission $|t_{1}|^2$ 
for resistive checkerboard metasurfaces with $Z\sub{c}=0,\, Z_0/2$, and $\infty$.
The frequency in the horizontal axis is normalized by $f_0:= c_0/a$ ($c_0$ is the speed of light in a vacuum). 
Above the frequency $f_0$, 
diffracted waves can propagate in free space.
The checkerboard metasurfaces with $Z\sub{c}=0$ and $\infty$ resonate at the same frequency $f/f_0=0.89$.
Babinet's principle assures that the sum of these 
transmission spectra equals 1 in the region of $f\leq f_0$,
because the checkerboard metasurface with $Z\sub{c}=\infty$
is complementary to that with $Z\sub{c}=0$.
For the resistive checkerboard metasurface with $Z\sub{c}=Z_0/2$, 
transmission equals to $1/4$ independent of frequency, 
even when diffraction takes place ($f\geq f_0$).
This constant response seems very 
strange, because metasurfaces made from metal usually 
exhibit a resonant response,
but it can be explained by Theorem~\ref{theo:5}.
In addition to the magnitude of transmission,
we also confirm the phase of $t_{1}$.
For $0<f/f_0\leq 1.25$,
we have $|\Im [t_{1}]/\Re [t_{1}]|<1.2\times 10^{-2}$.
This result shows $t_{1}=1/2$ expected by Theorem~\ref{theo:5}.

The right panel of Fig.~\ref{fig:simulation}(a) shows the spectrum of energy loss 
for the resistive checkerboard metasurface with $Z\sub{c}=Z_0/2$.
The loss is calculated by integration of the Poyinting vector 
over the resistive sheets.
In the frequency range $f\leq f_0$,
we can see that half the incident power is absorbed by the metasurface,
while the electromagnetic energy is converted to diffracted modes
in $f\geq f_0$. These results agree with Theorem~\ref{theo:5},
and coherent perfect absorption can be realized for the two-side excitations.
Perfect absorption occurs for any $d/a$.
The resistive checkerboard metasurface with tiny resistive sheets
can absorb energy in very small regions.
This property can be useful for the enhancement of non-linearity of resistance.

The resistive checkerboard metasurface with $Z\sub{c}=Z_0/2$ also has
translational self-complementarity.
Then, it exhibits frequency-independent response for oblique 
incident waves.
By using \textsc{HFSS}, we calculated the response of the resistive checkerboard metasurfaces with $Z\sub{c}=Z_0/2$ for an oblique 
incidence of a circularly polarized plane wave with incident angle $45^\circ$ in the $xz$-plane.
In this case, we obtain the same transmission spectra for 
the right and left circularly polarized incident waves.
\footnote{The resistive checkerboard metasurfaces on $z=0$ are invariant 
when we perform the rotation by $180^\circ$ about $x$ axis after the mirror reflection $z\rightarrow -z$. Because the helicity of an incident wave is changed under this operation, 
$t\sub{RR}=t\sub{LL}$ and $|t\sub{RL}|=|t\sub{LR}|$ are derived.
}
The obtained spectra of $|t_{1}|^2$ and $|t_{2}|^2$ are shown in Fig.~\ref{fig:simulation}(b).
We can see that $|t_{1}|^2=1/4$, while $|t_{2}|^2$ has two non-zero resonant peaks at $f/f_0=0.58$ and $1.02$.
Slight changes of $|t_{1}|^2$ are considered as numerical errors.
For $0<f/f_0\leq 1.25$, we have $|\Im [t_{1}]/\Re [t_{1}]|<1.7\times 10^{-2}$.
This result supports the validity of Theorem~\ref{theo:7}.
The two peaks of $|t_2|^2$ are originated from 
the interaction between lattice sites.
Periodic systems exhibit such singular behaviors 
when a diffracted beam grazes to the plane $z=0$ (Rayleigh condition),\cite{GarciadeAbajo2007}
and in our system, the Rayleigh condition is satisfied at $f/f_0=1$ and $f/f_0= 2-\sqrt{2}=0.586$.
These frequencies correspond to 
the peaks shown in the graph.
Then, $|t_2|^2$ shows resonant behaviors near these frequencies,
while $|t_1|^2$ should be constant because of translational self-complementarity.

\section{\label{sec:6}Summary}
In this paper, we analyzed theoretically
electromagnetic plane-wave scattering by self-complementary metasurfaces.
In order to study the response of self-complementary metasurfaces,
we first described the electromagnetic duality and Babinet's principle with resistive elements.
Next, by applying Babinet's principle, we obtained 
the relation of scattering coefficients for a metasurface and its complement.
Using this result,
we revealed that the frequency-independent transmission and reflection are realized for self-complementary metasurfaces. In the case of normal incidence of a circularly polarized plane wave onto a self-complementary metasurface, 
complex transmission and reflection coefficients of the 0th-order  diffraction must be $1/2$ and $-1/2$, respectively.
If a self-complementary metasurface additionally has $n$-fold rotational symmetry $(n\geq 3)$, 
the above result is valid for normal incidence of a plane wave with an arbitrary polarization.
Furthermore, we found that this type of metasurface acts as a coherent perfect absorber.
We also studied metasurfaces with translational self-complementarity.
For an oblique incidence of a circularly polarized plane wave to a metasurface with 
translational self-complementarity, 
complex transmission and reflection coefficients of the 0th diffraction order also equal to $1/2$ and $-1/2$, respectively.
These results are confirmed by numerical simulations for resistive checkerboard metasurfaces.

\begin{acknowledgments}
The authors would like to thank M.~Hangyo and S.~Tamate for fruitful discussions,
T.~Kobayashi for his support in numerical simulations,
and T.~McArthur for his helpful comments.
Y.~Terekhov, S.~M.~Barnett, and R.~C.~McPhedran are also gratefully acknowledged for 
giving us information on several papers.
This work was support
in part by Grants-in-Aid for Scientific Research Nos.~22109004 and 
25790065. YN acknowledges support from the Japan Society for the Promotion of Science.
\end{acknowledgments}

\appendix 
\section{\label{append1} The relation between totally transmitted and totally reflected waves}
We consider an incident wave $(\vct{E}\sub{in},\vct{H}\sub{in})$ in $z\leq 0$ 
and the totally transmitted wave $(\vct{E}\sub{TT},\vct{H}\sub{TT})$ in $z\geq 0$.
If there is a surface made of PEC on $z=0$, 
the incident wave is totally reflected.
This totally reflected wave is denoted by $(\vct{E}\sub{TR},\vct{H}\sub{TR})$.
We show that $(\vct{E}\sub{TR},\vct{H}\sub{TR})$ can be represented by
 $(\vct{E}\sub{TT},\vct{H}\sub{TT})$ like the method of images
used in electrostatics.
We define $\mathcal{M}$ as the mirror reflection with respect to $z=0$.
If we assume
\begin{equation}
 (\vct{E}\sub{TR},\vct{H}\sub{TR}) 
= - (\mathcal{M}\vct{E}\sub{TT},\mathcal{M}\vct{H}\sub{TT}),  \label{eq:13}
\end{equation}
the boundary condition of perfect electric conductor
is satisfied. This is because $\mathcal{P} (\vct{E}\sub{in}+\vct{E}\sub{TR})=\mathcal{P} (\vct{E}\sub{in}-\mathcal{M}\vct{E}\sub{TT})=0$ for $z=0$.
Then, the definition of Eq.~(\ref{eq:13}) is valid.
Because magnetic fields are axial vectors,
$\mathcal{P}\mathcal{M}\vct{H}\sub{TT} =  -\mathcal{P}\vct{H}\sub{in}$ 
is satisfied on $z=0$. From this equation and Eq.~(\ref{eq:13}),
\begin{equation}
 \mathcal{P}\vct{H}\sub{TR} = \mathcal{P} \vct{H}\sub{in}  \label{eq:14}
\end{equation}
is satisfied for $z=0$.

\section{\label{append2} Relation of scattering coefficients for all diffracted components}
We generalize Theorem~\ref{theo:2} to include all diffracted modes.
The two problems discussed in Sec.~\ref{sec:3} are considered.
For $(p,q)=(0,0)$, we define $(\tilde{\vct{E}}^+_{pq,1},\tilde{\vct{H}}^+_{pq,1}):=
(\tilde{\vct{E}}_{\mathrm{TT}}, \tilde{\vct{H}}_{\mathrm{TT}})$ 
and its perpendicular polarization state $(\tilde{\vct{E}}^+_{pq,2},\tilde{\vct{H}}^+_{pq,2})$. For $(p,q)\ne (0,0)$, we also define 
$(\tilde{\vct{E}}^+_{pq,1},\tilde{\vct{H}}^+_{pq,1})$
and $(\tilde{\vct{E}}^+_{pq,2},\tilde{\vct{H}}^+_{pq,2})$
that are two orthogonal polarized modes with the factor 
$\ee^{\ii \underline{\vct{k}}_{pq}\cdot \vct{x}} \ee^{\ii k_z z}$,
where $k_z = \sqrt{ 
|\vct{k}_0|^2-|{\underline{\vct{k}}_{pq}}|^2 
}$ ($\Im k_z \geq 0$).
The mirror symmetric fields of 
$(\tilde{\vct{E}}^+_{pq,\alpha}, \tilde{\vct{H}}^+_{pq,\alpha})$ 
with respect to $z=0$ are denoted 
by $(\tilde{\vct{E}}^-_{pq,\alpha}, \tilde{\vct{H}}^-_{pq,\alpha})$.

We then decompose the complex field of problem~(a) in $z\geq 0$
as $$\sum_{(p,q)\in \mathbb{Z}^2}\sum_{\alpha=1,2} t_{pq, \alpha}
(\tilde{\vct{E}}^+_{pq,\alpha}, \tilde{\vct{H}}^+_{pq,\alpha})$$
with complex transmission coefficients $t_{pq,\alpha}$.
In $z\leq 0$, the field is given by 
$$(\tilde{\vct{E}}\sub{in}, \tilde{\vct{H}}\sub{in})+\sum_{(p,q)\in \mathbb{Z}^2}\sum_{\alpha=1,2} r_{pq, \alpha}
(\tilde{\vct{E}}^-_{pq,\alpha}, \tilde{\vct{H}}^-_{pq,\alpha}).$$
For problem (b), we define
 $(\tilde{\vct{E}}'^\pm_{pq,\alpha}, \tilde{\vct{H}}'^\pm_{pq,\alpha}):=
\mathcal{R}_{\mp\pi/2} (\tilde{\vct{E}}^\pm_{pq,\alpha}, \tilde{\vct{H}}^\pm_{pq,\alpha})$.
The fields in problem (b) are represented as follows:
 $$\sum_{(p,q)\in \mathbb{Z}^2}\sum_{\alpha=1,2} t'_{pq,\alpha}(\tilde{\vct{E}}'^+_{pq, \alpha}, \tilde{\vct{H}}'^+_{pq, \alpha})$$
in $z \geq 0$, and 
$$(\tilde{\vct{E}}'\sub{in}, \tilde{\vct{H}}'\sub{in})+\sum_{(p,q)\in \mathbb{Z}^2}\sum_{\alpha=1,2}r'_{pq,\alpha}(\tilde{\vct{E}}'^-_{pq, \alpha}, \tilde{\vct{H}}'^-_{pq, \alpha})$$ 
in $z\leq 0$.
Now, we can generalize Theorem~\ref{theo:2} as follows:
\begin{theo}
\label{theo:8}
$t_{00,1}+t_{00,1}'=1$,
$r_{00,1}+r'_{00,1}=-1$,
and 
$t_{pq,\alpha}+t'_{pq,\alpha}=0$,
$r_{pq,\alpha}+r'_{pq,\alpha}=0$ for $(p,q,\alpha)\ne(0,0,1)$.
\end{theo}
The proof of Theorem~\ref{theo:8} is similar to that of Theorem~\ref{theo:2}.

\section{\label{append3} General order diffraction by 
metasurfaces with translational self-complementarity}
We discuss the general scattering components
of diffracted waves by metasurfaces with translational self-complementarity.
An oblique incidence of a circularly polarized plane wave is considered.
We define $W:=\{(p,q)\in \mathbb{Z}^2; |\underline{\vct{k}}_{pq}| < |\vct{k}_0|\}$.
For $(p,q)\in W$, the waves with the wave vector 
$\underline{\vct{k}}_{pq}\pm  \sqrt{ 
{k_0}^2-|{\underline{\vct{k}}_{pq}}|^2 
}\,\vct{e}_z$ are not 
evanescent but propagating plane waves.
$(\tilde{\vct{E}}^+_{00, 1}, \tilde{\vct{H}}^+_{00,1})$
represents the totally transmitted plane wave with circular polarization.
For $(p,q)\in W$ satisfying $(p,q)\ne (0,0)$,
$(\tilde{\vct{E}}^+_{pq, 1}, \tilde{\vct{H}}^+_{pq,1})$
are selected to be the circularly polarized plane waves with the same helicity of 
$(\tilde{\vct{E}}^+_{00, 1}, \tilde{\vct{H}}^+_{00,1})$.
Now, Theorem~\ref{theo:7} is extended as follows:
\begin{theo}
\label{theo:9}
For the 0th-order modes, we have 
$t_{00,1}=t'_{00,1}=1/2$ and $r_{00,1}=r'_{00,1}=-1/2$.
For  $(p,q)\in W$ satisfying $(p,q)\ne (0,0)$, 
we have $t_{pq,1}=t'_{pq,1}=0$ and $r_{pq,1}=r'_{pq,1}=0$.
\end{theo}
This theorem is proved in the same manner as Theorem~\ref{theo:7}.
The latter part of this theorem shows that 
helicities must be converted for
propagating waves with $(p,q)\ne (0,0)$ diffracted 
by metasurfaces with translational self-complementarity.

\end{document}